\documentclass[conference]{IEEEtran}
\IEEEoverridecommandlockouts
\usepackage{cite}
\usepackage{amsmath,amssymb,amsfonts}
\usepackage{algorithm}
\usepackage{float}
\usepackage{algcompatible}
\usepackage{graphicx}
\usepackage{textcomp}
\usepackage{xcolor}
\usepackage{mathtools}
\usepackage{dsfont}
\usepackage{amsthm}
\usepackage{multicol}
\usepackage{caption}
\usepackage{comment}
\usepackage{url}
\usepackage{subcaption}
\usepackage{caption}
\newtheorem{theorem}{Theorem}
\newtheorem{corollary}{Corollary}

\def\BibTeX{{\rm B\kern-.05em{\sc i\kern-.025em b}\kern-.08em
    T\kern-.1667em\lower.7ex\hbox{E}\kern-.125emX}}
    
\newcommand{\note}[1]{{\color{red} #1}}


\begin{document}
\bstctlcite{IEEEexample:BSTcontrol}

\title{Constrained Deep Reinforcement Learning for Fronthaul Compression Optimization\\
\thanks{ This work has received funding from the European Union's Horizon Europe research and innovation programme under the Marie Skłodowska-Curie grant agreement No. 101073265.}
}

\author{\IEEEauthorblockN{Axel Gr\"onland\textsuperscript{1,2}, Alessio Russo\textsuperscript{1}, Yassir Jedra\textsuperscript{1}, Bleron Klaiqi\textsuperscript{2}, Xavier Gelabert\textsuperscript{2}}
\IEEEauthorblockA{\textit{\textsuperscript{1}Royal Institute of Technology (KTH), Stockholm, Sweden} \\ \textit{\textsuperscript{2}Huawei Technologies Sweden AB, Stockholm Research Centre, Sweden} \\
\{gronland, alessior, jedra\}@kth.se, \{bleron.klaiqi, xavier.gelabert\}@huawei.com }
}

\maketitle

\begin{abstract}
In the Centralized-Radio Access Network (C-RAN) architecture, functions can be placed in the central or distributed locations. This architecture can offer higher capacity and cost savings but also puts strict requirements on the fronthaul (FH). Adaptive FH compression schemes that adapt the compression amount to varying FH traffic are promising approaches to deal with stringent FH requirements.  
In this work, we design such a compression scheme using a model-free off policy deep reinforcement learning algorithm which accounts for FH latency and packet loss constraints. Furthermore, this algorithm is designed for model transparency and interpretability which is crucial for AI trustworthiness in performance critical domains. We show that our algorithm can successfully choose an appropriate compression scheme while satisfying the constraints and exhibits a roughly 70\% increase in FH utilization compared to a reference scheme.
\end{abstract}

\begin{IEEEkeywords}
C-RAN, fronthaul, machine learning, reinforcement learning, performance evaluation.
\end{IEEEkeywords}

\section{Introduction}
\label{sec:introduction}
Centralized Radio Access Network (C-RAN) allow the splitting of RAN functionalities between remote radio units (RRU) near antenna sites and the baseband unit (BBU) at a centralized location.
In this setup, a centralized pool of BBUs can jointly process RAN functions from multiple RRUs, allowing better resource utilization and dimensioning than non-centralized RAN options 
 \cite{CRAN_ChinaMobile}. C-RAN also offers increased maintainability, flexibility, upgradability, and improved coordination features such as coordinated multi-point (CoMP) and inter-cell interference coordination (ICIC) among others \cite{checko_cloud_2015}.

Conversely, C-RAN deployments may cause high data rate requirements on the fronthaul (FH) and increased latency in the signal processing chain \cite{checko_cloud_2015}. A major challenge in C-RAN deployments is the huge demands on bandwidth aggregation required for the FH, especially for specific split options (see e.g. \cite{duan_performance_2016}). 
Fortunately, one can resort to more favourable splits \cite{rodriguez_cloud-ran_2020}, and data compression methods \cite{lorca_lossless_2013, lagen_fronthaul-aware_2021,lagen_fronthaul_2022} that can help diminish the FH link data rate demands. Despite the use of these methods, further development is necessary and additional measures are required, as explained hereafter. 

Various approaches have been employed in recent studies to address the challenges mentioned above. For instance, \cite{liu_graph-based_2015} introduces a graph-based framework that effectively reduces the FH cost by appropriately splitting and placing baseband processing functions within the network.
A lossless FH compression technique (FH) is introduced in \cite{lorca_lossless_2013}, which relies on the proportion of utilized resources.
In \cite{lagen_modulation_2021}, authors offer insights into FH compression using modulation compression, with a reported reduction in required FH capacity of up to 82\%.
Modulation compression and scheduling strategies are combined in \cite{lagen_fronthaul_2022} to further optimize the use of FH-limited deployments. In \cite{Kang_FH_Comp_Prec_2016}, joint FH compression and precoding design is proposed, where two different splits are investigated to determine the location of the precoder.
Noteworthy, the works mentioned above rely on conventional mathematical optimization approaches, which are generally complex and require the availability of underlying models. However, obtaining optimal solutions can be challenging due to their high complexity and the difficulty of acquiring underlying models for realistic scenarios.
To help with this, Machine Learning (ML) approaches can be employed. The literature extensively covers the use of ML techniques for complex optimization problems in wireless networks, see e.g. \cite{ali_6g_2020} and sources cited therein. In the C-RAN domain, supervised learning has been used to investigate UE slicing and functional split optimization \cite{Matoussi_dl_ue_slice_2020}. However, obtaining high-quality labels for supervised learning can be challenging and costly in practice. Another proposed framework, in \cite{murti_learning-based_2022}, uses model-free reinforcement learning (RL) to jointly select the most suitable split and computing resource allocation.


In this work, we discuss the configuration and optimization of the FH in the downlink (DL) direction assuming a C-RAN deployment. We propose an off-policy constrained reinforcement learning framework. We show that we can partition value iteration which is useful for distribution and explainability. We implement two classical Deep RL algorithms which dynamically adjusts various parameters associated with FH compression schemes, such as modulation order, precoder granularity and precoder weight quantization. The primary objective is to maximize FH utilization, and thus the air interface throughput, while ensuring that FH latency and the packet loss rate remains below predefined thresholds. Our contribution is unique as it addresses the problem of FH compression optimization using learning-based techniques under FH latency and packet loss constraints. 

The remainder of this paper is organized as follows. In Section \ref{sec:background} we introduce the framework of Markov Decision Processes (MDPs) and Reinforcement Learning (RL). In Section \ref{sec:system_models} we describe the assumed C-RAN scenario, outline our simulation setup, and provide a description of our system models. The formulation of optimization problem is given in Section \ref{sec:prob_formulation}. We derive a RL problem formulation and show its equivalence to the original problem. Next, we derive a reward function that maps to the problem we are solving, along with some additional techniques to stabilize and speed up learning. In Section \ref{sec:numerical_evaluation}, we present the results of our experiments, focusing on the expected throughput gain compared to other less dynamic policies. Lastly, Section \ref{sec:conclusions} provides an overview of the final remarks.

\section{Background}\label{sec:background}
In this section we describe the framework of Markov Decision Processes and Reinforcement Learning.

\subsection{Markov Decision Process}\label{sec:background-MDP}


A Markov Decision Process  (MDP) is described by a tuple $(S,A, P, r,p_0)$ \cite{puterman2014markov}: $S$ is the state-space, representing all possible states that the system can be in; $A$ is the action-space, representing the set of all actions that an agent can take; $P: S\times A \times S\to[0,1]$ is the transition kernel, where  $P(s'|s,a)$ corresponds to the probability of transitioning from state $s$ to state $s'$ upon selecting action $a$;  $r: S\times A \to \mathbb{R}$ is the reward function, providing a scalar reward after taking an action in a given state; lastly, $p_0: S\to[0,1]$ denotes the initial state distribution of the MDP.
The system evolves over discrete time steps. At each step $t$, the MDP is in a state $s_t$, which is observed by an agent. The agent selects an action $a_t\in A$, according to a stationary Markov policy $\pi: S\times A \to [0,1]$, influencing the state transition governed by $P$.
The agent's objective is to find a policy $\pi$ that maximizes the total expected discounted reward over an infinite horizon. This is quantified by the value function under policy $\pi$, denoted as $V^\pi(s)$, and defined as:
$V^\pi(s) \coloneqq \mathbb{E}^\pi\left[\sum_{t=0}^{\infty} \gamma^t r(s_t, a_t)|s_0=s \right]$ for some discount factor $\gamma \in (0,1)$ and starting state $s\in S$. Similarly, we define the $Q$-value function of a policy as $Q^\pi(s,a)=\mathbb{E}^\pi\left[\sum_{t=0}^{\infty} \gamma^t r(s_t, a_t)|s_0=s,a_0=a \right]$. Finally, we denote by $\pi^\star(s) = \arg\max_\pi V^\pi(s)$ the optimal policy (or greedy policy).

\subsection{Reinforcement Learning}
Reinforcement Learning (RL) deals with the problem of learning an optimal policy $\pi$ in an MDP with unknown dynamics or large dimensional state-action spaces \cite{sutton2018reinforcement}. Off-policy methods like $Q$-learning \cite{watkins1992q} can be used for MDPs with finite state-action spaces to learn the greedy policy $\pi^\star$. However, in the context of high-dimensional state and action spaces, traditional tabular methods fall short. Deep RL algorithms, such as Deep $Q$-Networks (DQN) and Soft Actor-Critic (SAC), leverage neural networks to approximate the value function  $V^\pi$ and/or the policy $\pi$.  DQN \cite{mnih2015human}, an extension of the $Q$-Learning algorithm,  uses a deep neural network to approximate the $Q$-function $Q^\pi$, enabling RL in environments with continuous state-spaces. and finite-action spaces. SAC \cite{haarnoja2018soft} is an actor-critic algorithm, based on $Q$-learning, designed for environments with continuous action spaces, incorporating an entropy term to achieve a balance between exploration and exploitation. We refer the reader to \cite{mnih2015human,haarnoja2018soft} for more information regarding these algorithms.

\section{Problem Formulation and Model}
\label{sec:system_models}

In this section, we first present a formal description of the FH compression optimization problem in C-RAN under latency constraints. Then, we formulate this problem as a constrained MDP.

\subsection{C-RAN scenario description}\label{sec:Scenario_Desc}
It is assumed a 3GPP NR Time Division Duplex (TDD) system, with data flowing in the DL direction. In the frequency domain, Orthogonal Frequency Division Multiplexing (OFDM) subcarriers are separated by given subcarrier spacing (SCS), defined as $\Delta f_{\textsc{scs}}=15\cdot2^\mu$ (kHz) with SCS index $\mu=\{0,1,2,3,4\}$. The total available bandwidth $B$ (Hz) is divided into  $N^{\textsc{prb}}_{\textsc{B},\mu}$ Physical Resource Blocks (PRBs), each PRB containing 12 consecutive subcarriers. In the time domain, transmission of UL and DL data is cyclically carried out over a predefined time duration measured in UL and DL \emph{slots} respectively, following a pattern given by the configured TDD frame-structure. Each slot has a duration of 14 symbols, that is $T_{slot}^\mu=14\cdot T_{symb}^\mu $, with $T_{symb}^\mu$ the duration of an OFDM symbol. It is also common to define the number of subcarrier-symbol pairs contained in a single PRB during the duration of a slot, namely Resource Elements (REs), hence $N_{\textsc{re}}=12\cdot14=168$. Massive MIMO with $N_{ant}$ antennas at the transmitter is considered, where digital precoding is applied with pre-calculated weights based on channel estimations from UL pilot measurements \cite{khorsandmanesh_quantization-aware_2022}. As a result, up to $\upsilon_{lay}\le N_{ant}$ spatially multiplexed users (or layers) can be scheduled at the same time, over the same PRB.

In this study, we focus on the C-RAN architecture \cite{CRAN_ChinaMobile}, which allows the centralization of the baseband processing for multiple ($K$) geographically distributed RRUs, each serving a cell. The processing can be split between a centralized location housing a pool of BBUs and said individual RRUs.
Fig. \ref{fig:scenario}(a) illustrates the physical architecture for a simple $K=3$ cell scenario, where the RRUs and the BBU pool are interconnected via the FH. The FH provides a link of capacity of $C_{\textsc{fh}}$ (Gb/s), and transports the aggregate DL generated traffic towards the RRUs.  
From a logical architecture perspective, see Fig. \ref{fig:scenario}(b), and for each cell served by a RRU, we can consider the Baseband Low (BBL) entity, residing at each RRU, and the Baseband High (BBH) entity residing at the BBU pool. The term \emph{split} will be used hereafter to describe the amount of baseband processing functions residing in the BBL and the BBH. 
\begin{figure}
    \centering
    \includegraphics[scale=0.85]{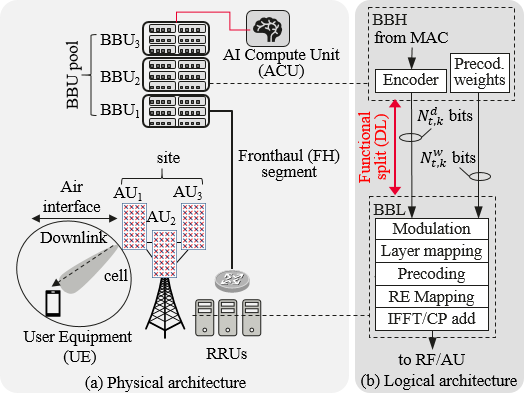}
    \caption{Considered scenario with (a) physical and (b) logical architectures.}
    \label{fig:scenario}
\end{figure}
Similar to \cite{lorca_lossless_2013, khorsandmanesh_quantization-aware_2022}, the adopted DL split considers the encoded user data and the precoding weights to be sent separately over the FH, see Fig. \ref{fig:scenario}(b). The encoded user data (in bits) can be mapped to the corresponding modulated symbols at the BBL, thus experiencing no quantization errors over the FH. The precoding weights on the other hand, being complex-valued samples, need to be quantized at a bit resolution $b^w$ (bits/sample), thus being prone to quantization errors. 
\subsection{FH compression configuration}\label{sec:Compression_config}
The FH capacity dimensioning usually takes into account the statistical multiplexing gain resulting from the spatial traffic distribution at different cells in a given area \cite{wangFronthaulStatisticalMultiplexing2017}. In practice, this means that the FH capacity is usually dimensioned below the maximum aggregated peak data rates across all cells, since the probability of such event is considered low. Nonetheless, in the event of high load across all cells, keeping the FH throughput below the (under-)dimensioned capacity will require compression methods. Specifically, three methods will be considered in this paper. Firstly, the modulation order, which in NR is given by the set $\mathcal{Q}^{\textsc{NR}}=\{2,4,6,8\}$ bit/symbol \cite{3gpp_ts_38.211_v15.5.0_nr}, can be restricted to some maximum value $q\in{\mathcal{Q}^{\textsc{NR}}}$ so that symbols are represented with fewer bits, thus lowering the FH throughput. This will, of course, have an impact on the air interface throughput, but its use can be limited to specific time instants where the FH cannot support the offered traffic. Second, we can modify the bitwidth $b^w\in\mathbb{N}$ (expressed in bits) for precoding weights used to quantify and represent complex-valued samples for their transmission over the FH. 
Finally, one can tune the sub-band precoding granularity, $r^w$, which reflects the number of subcarriers that can be assumed to share similar channel conditions and therefore admit the same precoding weight value \cite{lorca_lossless_2013}. We define the sub-band precoding granularity as the number of consecutive PRBs being applied the same precoding weight. For example, $r^w=1$ indicates that each PRB in the entire bandwidth will be precoded with its own distinctive weight, whereas $r^w=2$ means that two consecutive PRBs will be precoded with the same weight (effectively halving the weight payload to be transmitted over the FH), etc.

To ensure optimal air interface performance, it is crucial to employ FH compression only when necessary. This means using it when the FH utilization (the ratio between used and available FH capacity) is high and there is a risk of not delivering packets on time over the FH. At medium and low FH loads, compression should be reduced to increase FH utilization while maintaining performance below the limit. With this in mind, the main optimization criteria going forward is to maximize FH utilization while keeping FH latency and packet loss within acceptable limits. The latency for cell $k$, $L_k$, accounts for the queuing time in BBH and in the switcher as well as the transmission time over the FH link. We further denote the maximum latency as $L_{\max}$. In the event of high traffic the switcher may experience buffer overflow resulting in packets loss. We denote $LP_k$ as the packet loss for cell $k$ and aim to maintain the total packet loss over all cells $\sum_{k} LP_k = 0$.

\subsection{FH utilization}
Considering the functional split given in Sec. \ref{sec:Scenario_Desc}, the FH transports the aggregated traffic towards the different cells, comprising the data payload and precoding weights. The data payload (in bits) to be delivered at a given time slot $t$ intended for the cell $k$ is given by:
\begin{align}
N_{t,k}^d = N_{\textsc{re}}\cdot\upsilon_{lay}\cdot N_{t,k}^{\textsc{prb}}\cdot q_{t,k},
\end{align}
where $\upsilon_{lay}$, $N_{t,k}^{PRB}$, and $q_{t,k}$ denote number of layers, number of the allocated PRBs and modulation order, respectively.

In a similar way, the number of precoding weight bits transmitted at slot $t$ towards cell $k$ can be obtained as follows:
\begin{align}
N_{t,k}^w = \Biggl\lceil\frac{N_{t,k}^{\textsc{prb}}}{r_{t,k}^w}\Biggr\rceil\cdot \upsilon_{lay}\cdot N_{ant}\cdot b_{t,k}^w,
\end{align}
with $r_{t,k}^w$, $N_{ant}$, and $b_{t,k}^w$, the precoder granularity, the number of antennas, and the weight bit quantization, respectively.

The FH data rate at slot $t$ for cell $k$ is given by:
\begin{align}
R_{t,k}^{\textsc{fh}}= \frac{1}{T_{slot}^\mu}\left(N_{t,k}^d+N_{t,k}^w\right),
\end{align}
where $T_{slot}^\mu$ is the slot duration. The FH utilization is the ratio between the FH data rate and the FH capacity, $\rho_{t,k} = \frac{R_{t,k}^{\textsc{FH}}}{C_{\textsc{FH}}}$. We  abuse notation and also denote $\rho_t = \sum_k^K \rho_{t,k}$.

\subsection{Constrained MDP Formulation}\label{sec:III-MDPformulation}
Recalling that the compression configuration is controlled via three different parameters $(q_{k}, b_{k}^w, r_{k}^w)$,  we  define the system state at time $t$ to be: 
\begin{equation}\label{eq:state}
s_t=\{(\rho_{t, k},L_{t,k}, LP_{t,k}, q_{t,k}, b_{t,k}^w, r_{t,k}^w), k\in[K]\},
\end{equation}
which  includes also the FH utilization $\rho_{t, k}$, the FH latency $L_{t,k}$, the lost packets $LP_{t,k}$ and the configuration parameters $(q_{t,k}, b_{t,k}^w, r_{t,k}^w)$ at time $t$ and cell $k$.

The action at time $t$ is defined as an incremental change in configuration as follows:
\begin{equation}\label{eq:action}
a_t = \{(\Delta q_{t,k}, \Delta b_{t,k}^w, \Delta r_{t,k}^w ), k\in [K]\}.
\end{equation}
Here, $\Delta q_{t,k},\Delta b_{t,k}^w,\Delta r_{t,k}^w \in \{-1, 0, 1\}$ denote, respectively, changes in the parameters $q_{t, k}$, $b_{t,k}^w$, and $r_{t,k}^w$. More specifically, $\Delta q_{t,k} = -1$ leads to a lower modulation order, while $\Delta q_{t,k} = 1$ moves to higher modulation order  and $\Delta q_{t,k} = 0$ causes  no change. Same action encoding applies to $\Delta b_{t,k}^w$ and $\Delta r_{t,k}^w$. 

Note that the size of the joint action-spaces is exponential in $K$. Hence, for DQN we assume \emph{homogeneous load}, so that the action spaces is of constant size, i.e., $O(1)$. Alternatively, using actor-critic algorithms (such as SAC) it is possible to alleviate this issue by using an auto-regressive policy (see also the appendix\footnote{Technical report with appendix: \url{https://arxiv.org/abs/2309.15060}}).

Lastly, we define the reward and the constraints. The reward at time $t$ is defined as
\begin{equation}\label{eq:reward}
    r(s_t, a_t) = \rho_t.
\end{equation}
For the cell $k$, the latency $L_k$ and packet loss $LP_k$ constraints are introduced in the problem formulation as follows
\begin{align}
    \mathbb{P}_{s\sim d^\pi}\left(\max_k L_{k} > L_{max}\right) & < \xi, \label{eq:constraint_1} \\
    \mathbb{P}_{s\sim d^\pi}\left(\sum_k LP_{k} \neq 0 \right) & < \xi, \label{eq:constraint_2}
\end{align}
where $\xi$ is a prescribed confidence level and $d^\pi$ denotes the ergodic stationary measure under policy $\pi$.


\section{Method}\label{sec:prob_formulation}

In this section, we present RL methods to solve constrained MDPs . For clarity, we will consider a generic constrained RL problem modeled as an infinite horizon MDP with discounted factor $\gamma$: 
\begin{align}
\max_{\pi}\qquad &\mathbb{E}_{s\sim p_0}\left[\sum_t^\infty\gamma^tr(s_t, a_t) | s_0=s\right] ,\label{eq:Objective}\\
\textrm{s.t.} \qquad &\mathbb{P}_{s\sim d^\pi}\big(s\notin\mathcal{S}_i) < \xi_i,\qquad i=1,\dots,N, \nonumber\\
& s_{t+1} \sim P(\cdot|s_t,a_t), a_t\sim \pi(\cdot|s_t),\nonumber
\end{align}
where $\mathcal{S}_i$ and $\xi_i$ are the safety set and confidence level parameter for the $i$-th constraint, respectively. We remark that our constrained MDP formulation presented in section \ref{sec:III-MDPformulation} is an instance of the optimization problem defined in \eqref{eq:Objective}. Indeed, the states are given by \eqref{eq:state}, the actions are given by \eqref{eq:action}, the rewards are given by \eqref{eq:reward}, and the constraints are defined in \eqref{eq:constraint_1} and \eqref{eq:constraint_2}.


\subsection{Constrained Reinforcement Learning Approach}

\paragraph{Formulating the dual problem} 
Just as in previous works \cite{russo2022balancing,paternain2022safe}, we relax the problem by considering the constraints over the discounted stationary distribution of the states $d_\gamma^\pi(s) = (1-\gamma) \sum_{s'} p_0(s')\sum_{t\geq 0}\gamma^t \mathbb{P}(s_t=s|\pi,s_0=s')$, which  is the discounted probability of being in state $s$ at time $t$ when following policy $\pi$. In particular, when $\gamma\to 1$ we find the original optimization problem \cite{russo2022balancing,puterman2014markov}. Then, notice that $\mathbb{E}_{s\sim d_\gamma^\pi}[ \mathbf{1}_{\{s \notin S_i\}} ] = 1-(1-\gamma) \mathbb{E}_{s_0\sim p_0}^\pi[\sum_{t\geq 0} \gamma^t\mathbf{1}_{\{s_t \in S_i\}}]$. Let $\lambda=\begin{bmatrix} 1 & \lambda_1 &\dots &\lambda_N\end{bmatrix}^\top$ and 
\begin{align*}
    R(s,a) \coloneqq \begin{bmatrix}
    r_0(s,a) & r_1(s,a) &\dots &  r_N(s,a)
\end{bmatrix}^\top,
\end{align*} 
where
$r_0(s,a) = r(s,a)$ and $r_i=(1-\gamma)\mathbf{1}_{\{s \in S_i\}}$ for $i=1,\dots,N$. Thus, we can approximate the previous problem in  \eqref{eq:Objective} as an unconstrained optimization problem
\begin{equation}
    \min_{\lambda_1, \dots,\lambda_N>0}\max_{\pi}\mathbb{E}_{s\sim p_0}^\pi\left[\sum_t^\infty\gamma^t \lambda^\top R(s,a) \Big | s_0=s\right] + \lambda^\top \xi,
\end{equation}

where $\xi =\begin{bmatrix} 0 & \xi_1-1& \dots & \xi_N-1\end{bmatrix}^\top$. We observe that for a fixed vector $\lambda$ the inner maximization amounts to solving an RL problem. Consequently, we can see address the problem using a descent-ascent approach, where we use classical RL techniques to find an optimal policy that maximizes $\mathbb{E}_{s\sim p_0}^\pi\left[\sum_t^\infty\gamma^t \lambda^\top R(s,a) \Big | s_0=s\right]$, while updating the dual variables $\lambda_i$.

\paragraph{Value function decomposition} Notice that the inner problem is linear in the Lagrangian variables. Then, we can rewrite the  objective function as
\begin{equation}
    \min_{\lambda_1, \dots,\lambda_N>0}\max_{\pi} \lambda^\top V^\pi + \lambda^\top \xi,
\end{equation}
with $V^\pi = \begin{bmatrix}
    \mathbb{E}_{s\sim p_0}[V_0^\pi(s)] & \dots & \mathbb{E}_{s\sim p_0}[V_N^\pi(s)]
\end{bmatrix}^\top$ and $V_i^\pi(s) = \mathbb{E}^\pi\left[\sum_t^\infty\gamma^t r_i(s,a) \Big | s_0=s\right]$.
We propose to use RL to learn separately each term $V_i^\pi(s)$ (as well as the greedy policy of $\lambda^\top V^\pi$), instead of learning directly $\lambda^\top V^\pi$ and its greedy policy. Empirically, this seems to improve convergence. Practically, one may be interested in knowing the value for each separate reward function $r_i$ (and thus,  derive some insights  on the value of the constraints).
To that aim, we need a theorem that guarantees that updating each value function separately still leads to the same optimal policy $\pi^\star$ that maximizes $\mathbb{E}_{s\sim p_0}^\pi\left[\sum_t^\infty\gamma^t \lambda^\top R(s,a) \Big | s_0=s\right] $. We also let $Q_\lambda^\star$ be the $Q$-value function of the policy $\pi^\star =\arg\max_\pi \sum_i\lambda_i V_i^{\pi}$.

Then, let $K=\mathbb{R}^{(N+1)\times |S|\cdot|A|}_+$ be the set of real positive matrices of size $(N+1)\times |S|\cdot|A|$, and let $Q=\begin{bmatrix}Q_1 &\dots &Q_{N+1}\end{bmatrix}^\top\in K$, where each $Q_i$ is of size $|S|\cdot|A|$. Define ${\cal T}_{ \lambda}:K \to K$ to be the following  the Bellman operator
\begin{equation}
\resizebox{\hsize}{!}{%
${\cal T}_{ \lambda}  Q(s,a) \coloneqq  r(s,a) + \gamma \mathbb{E}_{s'}\left[ Q(s', \arg\max_{a'}  Q(s',a')^\top  \lambda) \right].$%
}
\end{equation}
where $Q(s,a)$ in this case is a row vector of size $N+1$ and $s'\sim P(\cdot|s,a)$. 
The following theorem guarantees that by using value-iteration we can indeed find the same optimal policy.
\begin{theorem}\label{theorem:convergence_value_iteration}
    There exists $Q^\star \in K$ satisfying the fixed point
\[
  Q^\star(s,a) =  r(s,a) + \gamma \mathbb{E}_{s'}\left[ Q^\star(s', \arg\max_{a'}  Q^\star(s',a')^\top  \lambda) \right].
\]
Then, value iteration using ${\cal T}_{ \lambda}$ converges to $Q^\star$, i.e., $\lim_{k\to\infty} Q_k =Q^\star$, where $Q_k = {\cal T}_\lambda^{k-1} Q_0$ for some $Q_0\in K$. Moreover, $Q^\star=\begin{bmatrix}Q_1^\star &\dots &Q_{N+1}^\star\end{bmatrix}$ satisfies $Q_\lambda^\star = \sum_i \lambda_i Q_i^\star$.
\end{theorem}
We refer the reader to the appendix for all the proofs.

\paragraph{Q-learning extension} By extension, the previous result also applies to standard $Q$-learning, as shown here.
\begin{corollary}\label{theorem:convergence_q_learning} Consider a  policy that visits all state-action pairs infinitely often, and assume to learn the $Q$-values according to the following update at time $t$ for each $i$
\[
\resizebox{\hsize}{!}{%
$Q_{t+1,i} (s_t,a_t) = (1-\alpha_t) Q_{t,i}(s_t,a_t) + \alpha_t (r_i(s_t,a_t) + \gamma  Q_{t,i}(s_{t+1}, a_t')),$%
}
\]
for some learning rate $\alpha_t$ satisfying the Robbins-Monro conditions \cite{borkar2009stochastic} and  $a_t' = \arg\max_{a'} \sum_i \lambda_i Q_{t,i}(s_{t+1}, a')$. Let $
Q_t = \begin{bmatrix}
    Q_{t,0}&
    \dots&
    Q_{t,N}
\end{bmatrix}^\top$.
Then, w.p.1, $\lambda^\top Q_t \to Q_\lambda^\star$.
\end{corollary}

\subsection{Algorithms}
\label{sec:alg}
In this section we outline how to change DQN and SAC to solve the optimization problem in (\ref{eq:Objective}).

\paragraph{DQN Extension} Based on the argument of  Corollary \ref{theorem:convergence_q_learning} , we propose a straightforward adaptation for DQN. We  initialize $(N+1)$ $Q$-networks  parameterized by $\theta_i, i=0,\dots,N$ (with the corresponding target networks $\bar \theta_i$). Then, the $i$-th parameter is updated according to $\theta_{i} \gets \theta_{i} + \alpha_t \nabla_{\theta_i} {\cal L}_{\theta_i}$, where
\[
\resizebox{\hsize}{!}{%
${\cal L}_{\theta_i} = \mathbb{E}_{(s,a,r,s') \sim {\cal B}}\left[(r_i+\gamma Q_{\bar \theta_i}(s',\pi_\theta(s')) - Q_{\theta_i}(s,a))^2\right],$%
}
\]
with ${\cal B}$ being the replay buffer. Then, we define the greedy policy according to $\pi_\theta(s) = \arg\max_{a} \sum_i \lambda_i Q_{ \theta_i}(s,a)$, and the overall value as $\sum_i \lambda_i Q_{\theta_i}$.

\paragraph{SAC Extension}
Also for SAC we propose a simple modification. The policy evaluation step follows directly from the DQN extension, while  for the policy improvement step
we minimize the following KL divergence:
\begin{equation*}
    \pi_\phi \gets \arg\min_\phi \mathbb{E}_{s\sim\mathcal{B}} \left[{\rm KL}\left(\pi_\phi(\cdot|s), \frac{\exp(\sum_i \lambda_i Q_{\theta_i}(s,\cdot))}{Z_\theta}\right)\right],
\end{equation*}
for some policy parametrized by $\phi$.

\paragraph{Optimization of the dual variables} As previously argued, we use a descent-ascent approach. This method entails solving an inner RL optimization problem, and then updating the Lagrangian variables accordingly. Therefore, once we have learnt value functions we can perform a gradient update in our dual variable. Then our Lagrange update will be similar to \cite{safeRL}:
\begin{equation}
    \lambda_{t+1} =\max(0, \lambda_t - \beta_t (V_{\theta} +\xi)),
\end{equation}
for some learning rate $\beta_t$ and vector $V_{\theta} = \begin{bmatrix}
    V_{\theta_0} & \dots V_{\theta_N}
\end{bmatrix}^\top$ and $V_{\theta_i}(s) = \max_{a} \sum_i \lambda_i Q_{\theta_i}(s,a)$  (note that the first component of $\lambda$ is set equal to $1$, and is not changed).
In practice, the descent-ascent approach can be done simultaneously by making sure that the learning rate $\beta_t$ is sufficiently small compared to the learning rate of the inner maximization problem.

\begin{figure*}[tp]
\begin{subfigure}{.5\textwidth}
\centering
  \includegraphics[width=.95\linewidth]{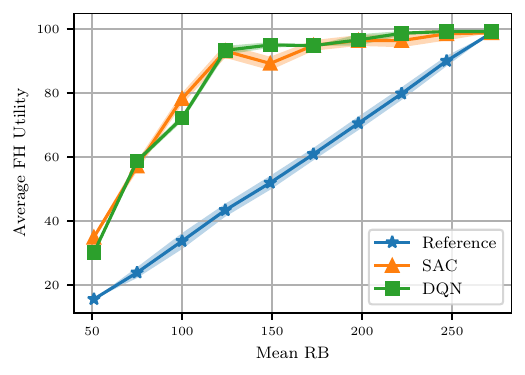}
  \caption{Average FH utilization over mean RB number}
  \label{fig:sub1}
\end{subfigure}%
\begin{subfigure}{.5\textwidth}
  \centering
  \includegraphics[width=.95\linewidth]{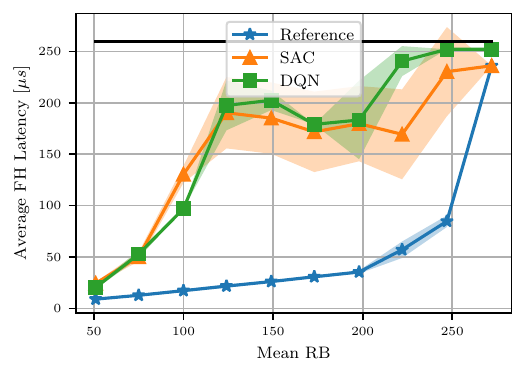}
  \caption{Average latency over mean RB number}
  \label{fig:sub2}
\end{subfigure}
\caption{Average FH utilization over the average RB number(averaged over cells). The black line in figure \ref{fig:sub2} shows the latency criterion which is set to 260 $\mu$s. Shaded regions show 3 sample standard deviations}
\label{fig:gain}
\end{figure*}

\section{Numerical Evaluation}
\label{sec:numerical_evaluation}

\noindent 
System simulations are performed via in-house build Baseb{A}nd {S}ystem {S}imulator (BASS) based on ns-3. BASS is used to simulate transmission of DL data and precoding weights over FH and measure latency and packet loss. We use $K=3$ cells that share a FH with capacity of $25$ Gbps.
The system parameters and values used for the simulations are shown in Table \ref{tab:Sim_Param}. For DQN we assumed \emph{homogeneous load}, to simplify the action space, while for SAC we used an auto-regressive policy (see also the appendix for details) to address the exponential size of the actions-space. 

\begin{table}[b]
\begin{center}
\caption{Simulation parameters.}
\begin{tabular}{ |c|c|c| }
 \hline
 {\bf Parameter} & {\bf Symbol} & {\bf Value} \\ 
 \hline
 Bandwidth & $B$ & 100 MHz\\
 \hline
 Number of available PRBs & $N^{\textsc{prb}}_{\textsc{B},\mu}$ & 273\\
 \hline
 Number of scheduled PRBs & $N_{t,k}^{\textsc{prb}}$ & 1\ldots273 \\
 \hline
 variance of scheduled PRBs & $\sigma_{N^{\textsc{prb}}}$ & 1 \\
 \hline
 Number of REs per RB & $N_{\textsc{re}}$ & $12\times14=168$\\
 \hline
 Subcarrier spacing index & $\mu$ & 1 ($\Delta f_{\textsc{scs}}=$30kHz)\\
 \hline
 Symbol duration & $T_{symb}^\mu$ & 33.33 $\mu$s\\
 \hline
 Slot duration & $T_{slot}^\mu$ & 0.5 ms\\
 \hline
 Number of cells & $K$ & 3\\
 \hline
 Number of antennas & $N_{ant}$ & 64\\    
 \hline
 Number of layers & $\upsilon_{lay}$ & 12\\ 
 \hline
 Modulation order & $q_{t,k}$ & $\mathcal{Q} =\{6,8\}$ \\    
 \hline
 Number of weight bits & $b_{t,k}^w$ & $\mathcal{B}^w =\{16,\ldots,22\}$\\    
 \hline
 Precoder granularity & $r_{t,k}^w$ & $\mathcal{R}^w =\{1,2,4\}$\\    
 \hline
 Max. latency safety parameter & $\xi$  & $0.025$\\
 \hline
 Max allowed latency & $\tau_{max}$  & 260 $\mu s$\\
 \hline
 FH capacity & $C_{\textsc{fh}}$ & 25 Gb/s\\    
 \hline

\end{tabular}
\label{tab:Sim_Param}
\end{center}
\end{table}

Fig. \ref{fig:sub1} shows the average FH utilization over the average number of PRBs for SAC, DQN, and the reference scheme designed to support worst case scenario, i.e., for 273 PRBs. The FH utilization increases with increasing number of scheduled PRBs until it reaches maximum number of PRBs. DQN and SAC improve average FH utilization by 70.3\% on average compared to the reference scheme. The average FH utilization of DQN and SAC converges to the reference scheme when approaching high load. SAC \cite{haarnoja2018soft} performs just as well as DQN. 

Fig. \ref{fig:sub2} show that the average latency stays within the limit. In this experiment our safety parameter $\xi$ was set to $0.025$, and three standard deviations correspond to the $0.15\%$-percentile, which far exceeds our designated safety parameter. 

In the Appendix, in Fig. \ref{fig:sac_training}   and \ref{fig:DQN_train},  we depict the values of  the constraints converge for $V_{C_1}$ and $V_{C_2}$,  showing that all the value functions eventually learn to be above the designated constraints. Furthermore, it can be seen that the Lagrange multipliers converge as well.
\section{Conclusions}

\label{sec:conclusions}
In this paper, we have investigated an adaptive FH compression scheme which can operate under latency and packet loss constraints. We have formulated this problem as a constrained optimization problem and designed an RL algorithm to solve it. Finding the exact solution is NP hard and requires accurate modelling, which is hard to obtain in real scenarios. Therefore, we proposed a Deep-RL approach to solve the constrained optimization problem, entailing a novel update for the $Q$-values of the policy.  Simulation results have shown that our method successfully learns a FH compression policy that maximizes FH utilization, while satisfying FH latency and packet loss constraints. On average, the FH utilization is improved by 70.2\% for both the simplistic DQN and the more advanced SAC approach. 

\bibliographystyle{IEEEtran}
\bibliography{ref}

\onecolumn
\subsection*{Appendix A: Proofs}
In this section we provide the proof of Theorem \ref{theorem:convergence_value_iteration} and Corollary \ref{theorem:convergence_q_learning}.
\begin{proof}[Proof of Theorem \ref{theorem:convergence_value_iteration}]
First, we denote the optimal value function as
\begin{align*}
V_\lambda^\star(s) &= \max_\pi \lambda^\top V^\pi(s)= \max_\pi \sum_i \lambda_i V_i^\pi(s),
\end{align*}
and similarly, we let $Q_\lambda^\star(s,a)$ be the corresponding optimal action-value function.
Remember the definition of ${\cal T}_\lambda$ for a fixed $\lambda$
\[
{\cal T}_{ \lambda}  Q(s,a) \coloneqq  r(s,a) + \gamma \mathbb{E}_{s'}\left[ Q(s', \arg\max_{a'}  Q(s',a')^\top  \lambda) \right].      
\]
and  define also the operator $\hat{\cal T}_\lambda Q = (\lambda^\top {\cal T}_\lambda Q)^\top$. 
Consequently, we find that
\begin{align*}
 \hat {\cal T}_{ \lambda}  Q(s,a)&=  \lambda^\top  r(s,a) + \gamma \mathbb{E}_{s'}\left[  \lambda^\top  Q(s', \arg\max_{a'}  Q(s',a')^\top  \lambda) \right],\\
&=  \lambda^\top  r(s,a) + \gamma \mathbb{E}_{s'}\left[\max_{a'}  \lambda^\top  Q(s', a') \right],
\end{align*}
where we used the property $\max_{a'}  \lambda^\top  Q(s', a') =\lambda^\top  Q(s', \arg\max_{a'}  Q(s',a')^\top  \lambda)$.

Consequently, we see that $Q_\lambda^\star$ is a fixed point for $\hat {\cal T}_{ \lambda}$. Therefore, it is sufficient to prove also ${\cal T}_{ \lambda}$ admits a fixed point $Q^\star$ that satisfies  $ \hat{\cal T}_\lambda Q_\lambda^\star = (\lambda^\top {\cal T}_\lambda Q^\star)^\top$.

We prove first that  ${\cal T}_{ \lambda}$ is a contraction. Consider the following norm $\|Q\|_\lambda = \max_{s,a} |\lambda^\top Q(s,a)| = \||\lambda^\top Q|\|_\infty$ (remember that $Q$ is an $N+1\times |S|\cdot|A|$ matrix). We can easily verify that this is a valid norm. Especially, since the reward is always positive, we note that $\|Q\|_\lambda=0$ implies $Q=0$, since $\max_{s,a} |\lambda^\top Q(s,a)|=\max_{s,a} |\sum_i \lambda_i Q_i(s,a)|$. Since  $\lambda_i>0$ and $Q_{i}(s,a)\geq 0$, it must follow that each $Q_i$ is $0$, leading to $Q=0$. Then, consider the following set of inequalities:
\begin{align*}
    |(\hat {\cal T}_{ \lambda} Q_1- \hat {\cal T}_{ \lambda} Q_2)(s,a)| &=  \gamma|\lambda^\top \mathbb{E}_{s'} [Q_1(s', \arg\max_{a'} Q_1(s',a')^\top  \lambda) - Q_2(s', \arg\max_{a'} Q_2(s',a')^\top  \lambda)  ]|,\\
    &=\gamma | \mathbb{E}_{s'} [\max_{a'} \lambda^\top Q_1(s',a') - \max_{a'} \lambda^\top Q_2(s',a')   ]|,\\
    &\leq \gamma  \mathbb{E}_{s'} [|\max_{a'} \lambda^\top Q_1(s',a') - \max_{a'} \lambda^\top Q_2(s',a')|   ],\\
    &\leq \gamma \max_{s',a'}|\lambda^\top (Q_1(s',a') - Q_2(s',a'))|,\\
     &\leq \gamma \|Q_1-Q_2\|_\lambda.
\end{align*}
Therefore, taking the supremum over $(s,a)$ on the l.h.s. we find
\begin{align*}
\max_{s,a}   |(\hat {\cal T}_{ \lambda} Q_1- \hat {\cal T}_{ \lambda} Q_2)(s,a)|&= \max_{s,a}   |\lambda^\top{\cal T}_{ \lambda} (   Q_1- Q_2)(s,a)|,\\
&=
\|  {\cal T}_\lambda Q_1- {\cal T}_\lambda Q_2\|_\lambda \leq \gamma \|Q_1-Q_2\|_\lambda.
\end{align*}
Therefore this mapping ${\cal T}_\lambda$ is a contraction in the $\|\cdot\|_\lambda$ norm, and the fixed point coincides with that of $\hat {\cal T}_{ \lambda}$. This property guarantees that repeated application of the operator will converge to a point $Q$ satisfying 
\[
\|{\cal T}_\lambda Q -Q\|_\lambda = 0
\]
Since this is a valid norm, the point satisfying $\|{\cal T}_\lambda Q -Q\|_\lambda = 0$ is unique, and we denote it by $Q^\star$. Consequently, it follows that $Q_\lambda^\star=\lambda^\top Q^\star = \hat {\cal T}_\lambda Q^\star$.  Thus  the greedy value satisfies
\[
Q_\lambda^\star = \sum_i \lambda_i Q_i^\star, \hbox{ where }  Q_i^\star(s,a) =  r_i(s,a) + \gamma \mathbb{E}[ Q_i^\star(s',\arg\max_{a'}  Q_\lambda^\star(s',a'))].
\]

\end{proof}

\begin{proof}[Proof of Corollary \ref{theorem:convergence_q_learning}]
 In the following, we consider a Markov policy $\pi$ that in a finite state-action space MDP guarantees that every state-action pair is visited infinitely-often. Classical $Q$-learning considers the following stochastic approximation scheme
\[
Q_{t+1} (s_t,a_t) = (1-\alpha_t(s_t,a_t)) Q_t(s_t,a_t) + \alpha_t y_t
\]
where $y_t= r(s_t,a_t) +  \gamma \max_{a'} Q_t(s_{t+1}, a')$ is the target value at time $t$.

For a suitable learning rate,  under a behavior policy $\pi$ we are guaranteed that $Q_t \to Q_\lambda^\star$ a.s., the optimal action-value function. Consequently, we obtain that $\lim_{t\to\infty} Q_t = \sum_i \lambda_i Q_i^\star$ due to the aforementioned discussion, where the greedy policy is defined as $\pi^\star(s) =\arg\max_{a} \sum_i\lambda_i Q_i^\star(s,a)$. 

By standard arguments, it follows immediately  that  updating the $i$-th value according to  the following update
\[
Q_{t+1,i} (s_t,a_t) = (1-\alpha_t(s_t,a_t)) Q_{t,i}(s_t,a_t) + \alpha_t y_{t,i},
\]
with 
$
y_{t,i} = r_i(s_t,a_t) + \gamma  Q_{t,i}(s_{t+1}, a_t')$,  with $a_t' = \arg\max_{a'} \sum_i \lambda_i Q_{t,i}(s_{t+1}, a'),
$ guarantees convergence to $Q_i^\star$. To see this, let
\[
Q_t = \begin{bmatrix}
    Q_{t,0}\\
    \vdots\\
    Q_{t,N}
\end{bmatrix}, \quad y_t = \begin{bmatrix}
    y_{t,0}\\
    \vdots\\
    y_{t,N}
\end{bmatrix}
\]
Then
\[
Q_{t+1} (s_t,a_t) =  Q_{t}(s_t,a_t) + \alpha_t ({\cal T}_\lambda Q_{t} - Q_{t} + w_{t}),
\]
with $w_t = y_t - {\cal T}_\lambda Q_t$. We easily see that $\mathbb{E}[w_t|{\cal F}_t]=0$, where ${\cal F}_t= \sigma(s_1,a_1,\dots,s_t,a_t)$. Thus, since we already know that ${\cal T}_\lambda$ is a contraction, with fixed point $Q^\star$, and all the quantities are bounded, then we are guaranteed that $Q_t\to Q^\star$ a.s., thus $\lambda^\top Q_t \to Q_\lambda^\star$ almost surely.
\end{proof}

\subsection*{Appendix B: Algorithm}
The following algorithm summarizes our training loop for our off policy learning algorithm. We have omitted the typical sampling/data collection steps for the sake of clarity. 
\begin{algorithm}
\caption{DRL-FC based on SAC with Actor-Replay buffer.}\label{alg:DDQN}
    \begin{algorithmic}[1]
    \REQUIRE Learning rates $\alpha,\beta$; Initial value $\lambda_0$
    \STATE $\textbf{Initialize}$ $\theta$, $\phi$, $\bar\lambda$
    \FOR{$t=0,1,\dots$}
        \STATE $(s, a, s') \sim \mathcal{D}_{p}$ \% Sample from the critic buffer
        \STATE $\theta_i \leftarrow \theta_i - \eta\nabla_{\theta_i}\mathcal{L}(\theta) \qquad\forall i = 0, 1, ... N$
        \STATE Update new priorities $p$ in the critic buffer $\mathcal{D}_{p}$
        \STATE Update new priorities $1/p$ in the actor buffer $\mathcal{D}_{1/p}$
        \STATE $(s, a, s') \sim \mathcal{D}_{1/p}$ \%Sample from the actor buffer
        \STATE $\theta_i \leftarrow \theta_i - \eta\nabla_{\theta_i}\mathcal{L}(\theta) \qquad\forall i = 0, 1, ... N$
        \STATE $\phi \leftarrow  \arg\min_\phi D_{KL}(\pi_\phi||e^{\lambda^\top Q}/Z_\theta)$
        \STATE $\lambda \leftarrow \max(0, \lambda - \eta_\lambda (V^\pi_\psi + \xi))$
        \STATE Update new priorities $p$ in the critic buffer $\mathcal{D}_{p}$
        \STATE Update new priorities $1/p$ in the actor buffer $\mathcal{D}_{1/p}$
        \STATE $(s, a, s') \sim \mathcal{D}_{U}$ \%Sample from the uniform buffer
        \STATE $\theta_i \leftarrow \theta_i - \eta\nabla_{\theta_i}\mathcal{L}(\theta) \qquad\forall i = 0, 1, ... N$
        \STATE $\phi \leftarrow  \phi - \eta \nabla_\phi D_{KL}(\pi_\phi||e^{\lambda^\top Q}/Z_\theta)$
        \STATE Update new priorities $p$ in the critic buffer $\mathcal{D}_{p}$
        \STATE Update new priorities $1/p$ in the actor buffer $\mathcal{D}_{1/p}$
        \STATE $\bar\psi\leftarrow(1-\kappa)\bar\psi + \kappa\psi$ 
        \STATE $\lambda \leftarrow \max(0, \lambda - \eta_\lambda (V^\pi_\psi + \xi))$
    \ENDFOR
    \STATE $\textbf{Return}$ $\phi$
    \end{algorithmic}
\end{algorithm}

We employ a series of modifications in our method:
\begin{itemize}
    \item We use \emph{prioritized experience replay} \cite{schaul2015prioritized}, as well as a modification to our buffer to also include sampling of experiences for the actor via \emph{inverse prioritization sampling}\cite{saglam2022actorreplay}, which helps training the actor only on experiences that have small TD-error. 
    \item We also have a shared neural network for all of our value objectives $V_i$ and $Q_i$. All value function attempt to model the same state distribution, $d^\pi(s)$, thus we can save computational time and memory with a multi-headed neural network. To make this framework fit our prioritization scheme we simply sum the PAL priorities of all value heads, we thus are forced to include the importance weights to correct our introduced bias. 
    \item For our DQN solution we utilize Boltzmann exploration and the Double DQN modification\cite{van2016deep}. 
    \item For our actor-critic solution we need to model the joint distribution over all actions $\pi(a_1, ... a_K | s)$. which is exponential in $K$. A common solution, then, is to use an auto-regressive approach. Moreover, since each agent (RRU) needs to collaborate with the other agents, we believe that a simple way to achieve collaboration is to generate the full action in an auto-regressive manner. That is, we use the current state $s$ to select action $a_1$, then we use $(s,a_1)$ to select $a_2$, etc.  This modification makes our policy recurrent in the action sequence, and conditions the action of the $k$-th agent based on the actions selected by the previous agents. An important matter regards the order of the actions chosen which we know from other modelling works can have a severe impact on the performance\cite{menick2018generating}, since our sequence is rather short we argue that this is not a big issue. We can use any recurrent model of choice and in this work we use a Transformer\cite{46201attention}. 
\end{itemize}

Finally, the training parameters used are summarized in the following table. We used PyTorch as our automatic differentiation library.
\begin{center}
\begin{tabular}{ |c|c|c| }
\hline
{\bf Parameter} & {\bf Symbol} & {\bf Value} \\ 
\hline
Discount factor & $\gamma$ & 0.95\\
\hline
Q learning rate & $\eta_Q$ & $10^{-3}$\\
\hline
policy learning rate & $\eta_\pi$ & $10^{-4}$\\
\hline
$\lambda $ learning rate & $\eta_\lambda$ & $10^{-4}$\\
\hline

Soft update frequency & $\kappa$ & $5\cdot10^{-3}$\\
\hline
Entropy Constraint & $\mathcal{H}_0$ &$ 0.2\log|\mathcal{A}|$\\
\hline
Number of parameters in NN &  &$1.8\cdot10^{5}$\\
\hline
\end{tabular}
\label{tab:Sim_Param}
\end{center}

\subsection*{Appendix C: Modelling User Behaviour as Random Process}
\noindent
We model the traffic over time as a stochastic process $\{\overline{N}^{\textsc{prb}}_{t,k}\}_{t\in T}$. We make an assumption that the number of users remain constant over all cells, since $\overline{N}^{\textsc{prb}}_{t,k}$ is roughly proportional to the number of users we create our random process as follows.
\begin{align*}
    \overline{N}^{\textsc{prb}}_{t+1,k} &= \overline{N}^{\textsc{prb}}_{t,k} + w_k\\
    \overline{N}^{\textsc{prb}}_{0,k} &= U(1, N^{\textsc{prb}}_{\textsc{B},\mu})\\
    w_k &\sim W(\overline{N}^{\textsc{prb}}_{t,k})\\
    \textit{s.t.} \qquad \sum_k w_k &= 0\\
    N_{low} &< \overline{N}^{\textsc{prb}}_{t,k} < N_{high}\qquad \forall t
\end{align*}
The choice of distribution $W$ is irrelevant, in order to enforce the constraint we can use rejection sampling. Rejection sampling is however a very slow method and doesn't lend well as $K$ increases. To implement an efficient sampling algorithm we introduce the following algorithm:
\begin{algorithm}[h]
\caption{Propagate stochastic process}
    \begin{algorithmic}[1]
    \REQUIRE $X$, $x_{min}$, $x_{max}$
    \STATE $\hat X \leftarrow X$
    \FOR{i = 1, 2, ... K}
        \FOR{j = i+1, i+2, ... K}
            \STATE $x_{low} = -\min(3, \hat X_{i} - x_{min}, x_{max} - \hat X_{j})$
            \STATE $x_{high} = \min(3, \hat X_{j} - x_{min}, x_{max} - \hat X_{i})$
            \STATE $\Delta x \sim U(x_{low}, x_{high})$
            \STATE $X_{i} = X_i + \Delta x$
            \STATE $X_{j} = X_j - \Delta x$
        \ENDFOR
    \ENDFOR
    \STATE \textbf{Return} $X$
    \end{algorithmic}
\end{algorithm}

\noindent
The exact form of this distribution is not to obvious to us, however in scenarios where the propagated vector $X$ has all elements far away from the minimum and maximum constraints, then we can see that $W$ is simply the sum of $K-1$ uniform distributions. When $K \rightarrow \infty$ then $W$ will approach a Gaussian according to central limit theorem. For our case when $K = 3$ we have that our $W$ distribution is a triangle distribution, if not accounting for the boundaries.

\subsection*{Appendix C: Values and Lagrange multipliers}

\begin{figure}[!h]
    \begin{subfigure}{0.5\textwidth}
        \includegraphics{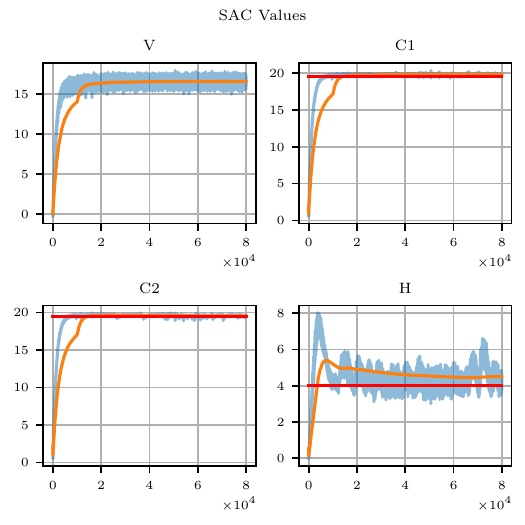}
        \label{fig:sac_values}
    \end{subfigure}
    \begin{subfigure}{0.5\textwidth}
        \includegraphics{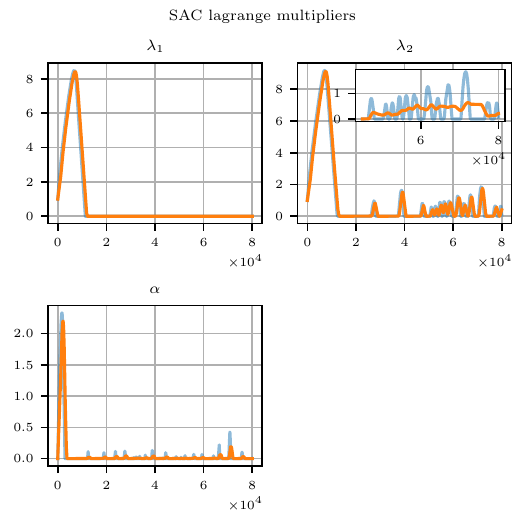}
        \label{fig:sac_multipliers}
    \end{subfigure}
    \caption{Training values and multipliers in over training iterations. we've also highlighted the running average in a window of 1000 iterations}
    \label{fig:sac_training}
\end{figure}
For SAC, in figure \ref{fig:sac_training} we can see that the values for the different value functions $V_r$, $V_{C_1}$, $V_{C_2}$ and $V_{\mathcal{H}}$. The red lines indicate where the constraints have been set to. We see that for all of the value heads we of our critic NN eventually learns to be above the designated constraints. We also see that the Lagrange multipliers in figure \ref{fig:sac_training} are oscillating very heavily, this can be explained by the introduced delay with exponential averaging of our target networks. For our Lagrange multipliers we essentially are tuning a control problem with large delays.


Similarly, we also see results for DQN in figure \ref{fig:DQN_train}. The convergence of the Lagrange multipliers in this case is a lot more apparent as we see that $\lambda_1 \rightarrow 0$ and $\lambda_2$ roughly approaches $0.14$. Just as before we've highlighted the constraints as red lines in the left plots of figure \ref{fig:DQN_train}. 

\begin{figure}[!h]
    \begin{subfigure}{0.5\textwidth}
        \includegraphics{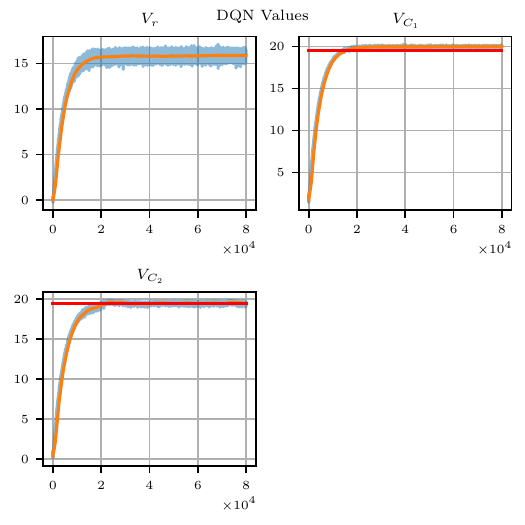}
        \label{fig:dqn_values}
    \end{subfigure}
    \begin{subfigure}{0.5\textwidth}
        \includegraphics{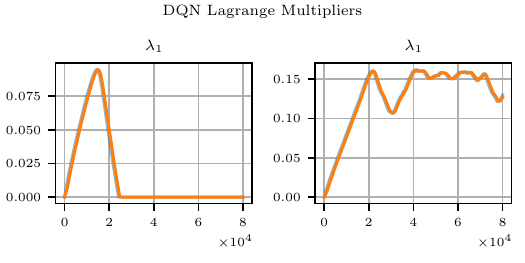}
        \label{fig:dqn_multipliers}
    \end{subfigure}
    \caption{Training values and multipliers in over training iterations. We've also highlighted the running average in a window of 1000 iterations}
    \label{fig:DQN_train}
\end{figure}

\end{document}